\documentclass[aps,pra,superscriptaddress,reprint]{revtex4-1}
\usepackage{amsmath}
\usepackage{amsthm}
\usepackage{amsfonts}
\usepackage{graphicx}
\usepackage[colorlinks=true]{hyperref}
\newtheorem{thm}{Theorem}
\newtheorem{lem}{Lemma}

\begin{document}
\title{The minimum error probability of quantum illumination}
\author{Giacomo De Palma}
\affiliation{QMATH, Department of Mathematical Sciences, University of Copenhagen, Universitetsparken 5, 2100 Copenhagen, Denmark}
\author{Johannes Borregaard}
\affiliation{QMATH, Department of Mathematical Sciences, University of Copenhagen, Universitetsparken 5, 2100 Copenhagen, Denmark}

\begin{abstract}
Quantum illumination is a technique for detecting the presence of a target in a noisy environment by means of a quantum probe.
We prove that the two-mode squeezed vacuum state is the optimal probe for quantum illumination in the scenario of asymmetric discrimination, where the goal is to minimize the decay rate of the probability of a false positive with a given probability of a false negative.
Quantum illumination with two-mode squeezed vacuum states offers a 6 dB advantage in the error probability exponent compared to illumination with coherent states.
Whether more advanced quantum illumination strategies may offer further improvements had been a longstanding open question.
Our fundamental result proves that nothing can be gained by considering more exotic quantum states, such as e.g. multi-mode entangled states.
Our proof is based on a new fundamental entropic inequality for the noisy quantum Gaussian attenuators.
We also prove that without access to a quantum memory, the optimal probes for quantum illumination are the coherent states.
\end{abstract}

\maketitle

\section{Introduction}
Quantum entanglement can increase both resolution and sensitivity in a number of metrological tasks \cite{giovannetti2004quantum,giovannetti2006quantum,giovannetti2011advances,Deganreview2017}.
By using highly entangled states such as NOON states~\cite{Demkowicz2015}, the sensitivity can in principle reach the Heisenberg limit yielding an improvement scaling with the square root of the number of probes compared to the standard quantum limit \cite{Mitchell2004,Afek2010,Facon2016}.
However, in most of these metrological tasks the improvement quickly vanishes in the presence of noise.
One remarkable exception is quantum illumination.

Quantum illumination was introduced in the setting of detecting the presence of a low-reflective target in a noisy environment using single photons as probes~\cite{lloyd2008enhanced}.
By entangling the signal photons with a quantum memory kept at the measurement station, it was shown that the initially strong quantum correlations between the signal photons and the memory resulted in a suppressed error probability compared to non-entangled signals.
Quantum Gaussian states \cite{braunstein2005quantum,weedbrook2012gaussian,serafini2017quantum} emerged as the prominent probes for quantum illumination.
Indeed, the coherent states perform better than the entangled single-photon states \cite{tan2008quantum,shapiro2009quantum}, and the two-mode squeezed vacuum states exhibit a further 6 dB decrease in the error probability exponent compared to the coherent states \cite{tan2008quantum,guha2009receiver,guha2009gaussian,zhuang2017optimum}.
The advantage of the two-mode squeezed vacuum states was experimentally verified in the optical regime \cite{lopaeva2013experimental,lopaeva2014detailed,zhang2015entanglement} and methods for extending this to the microwave regime have been proposed \cite{barzanjeh2015microwave}.
However, despite numerous theoretical studies \cite{wilde2017gaussian,zhang2014quantum,ragy2014quantifying,weedbrook2016discord,cooney2016strong,pirandola2011quantum, zhuang2017entanglement,las2017quantum,sanz2017quantum,paris2009quantum,pirandola2018fundamental} (see also the review \cite{genovese2016real}), determining whether the two-mode squeezed vacuum states are optimal or more advanced strategies may offer further improvements is still a longstanding open question.

\begin{figure} [t]
\includegraphics[width=0.48\textwidth]{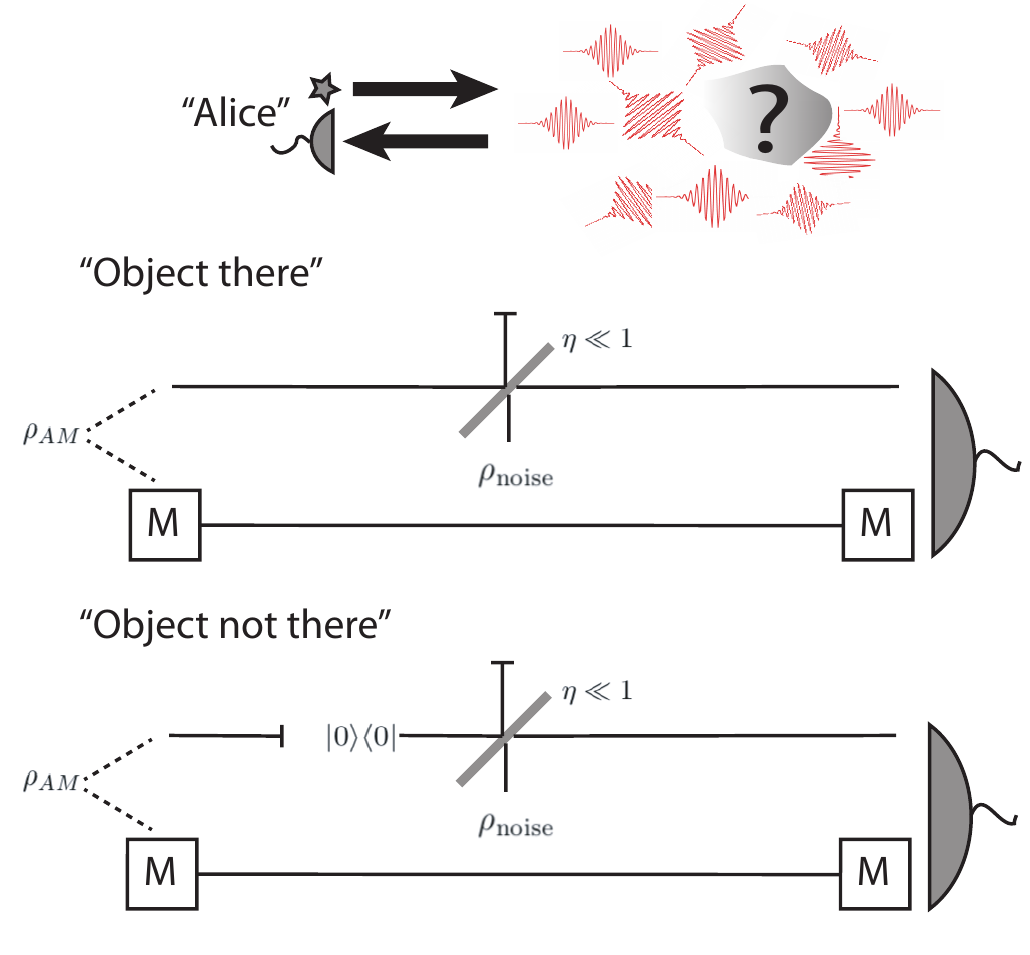}
\caption{The setup of quantum illumination.
Alice wishes to determine whether a low-reflective target is present in a noisy environment.
We model the setup as illustrated in the bottom of the figure, where Alice creates an arbitrary quantum state storing one part in a noise free memory and sending the other part as a signal.
The presence of the target can be modeled as a low transmission ($\eta\ll1$) beam-splitter with the signal and a strong thermal state as input while the absence of the target corresponds to replacing the signal input with the vacuum.
Upon discarding the non-detected modes, the beam-splitter with thermal noise corresponds to a noisy quantum Gaussian attenuator.}
\label{fig:figure1}
\end{figure}

In this paper, we answer this question by proving that the tensor products of identical two-mode squeezed vacuum states constitute the optimal probes among all the possible multi-mode probes in the scenario of asymmetric discrimination, where the goal is to minimize the decay rate of the probability of a false positive with a given probability of a false negative (Theorem \ref{thm:main}).
We also prove that coherent states constitute the optimal probe when a quantum memory is not available.
These two fundamental results completely solve the problem of the optimal probe for quantum illumination for asymmetric discrimination.
The striking and surprising consequence is that nothing can be gained with arbitrary many-mode entangled states compared to the two-mode squeezed vacuum state.

The proof of Theorem \ref{thm:main} is based on a new fundamental entropic inequality (Theorem \ref{thm:S}), stating that the tensor powers of the two-mode squeezed vacuum states minimize the quantum conditional entropy of the output of the noisy quantum Gaussian attenuators among all the input states with a bounded average number of photons.

\section{Quantum illumination}	
In the standard setup of quantum illumination (see \autoref{fig:figure1}), Alice wants to determine the presence of a low-reflective target in a noisy environment by sending an electromagnetic signal towards it.
If the target is present, Alice receives the reflected fraction of the signal plus the environmental noise.
If the target is not present, Alice receives just the noise.
Alice has then to distinguish between two quantum states: the received state with the reflected signal plus the noise and the received state with just the noise \footnote{We mention that also other scenarios have been considered in the literature.
In \cite{cooney2016strong}, the receiver is located on the other side of the target with respect to Alice, and he receives just the noise if the target is present, or the noise plus the transmitted signal if the target is not present.
More generally, one may consider a scenario of quantum reading \cite{pirandola2011quantum} where the receiver gets a reflection from a target or cell which is not necessarily faint.}.

Let us suppose that Alice wants to distinguish between the quantum state $\rho_0^{\otimes n}$ associated to the null hypothesis and the quantum state $\rho_1^{\otimes n}$ associated to the alternative hypothesis.
Let the POVM used for the discrimination have elements $\{M_n,\,\mathbb{I}_n-M_n\}$, where $0\le M_n\le\mathbb{I}_n$.
A type-I error or ``false positive'' occurs when Alice concludes that the state is $\rho_1^{\otimes n}$ when it actually is $\rho_0^{\otimes n}$, and has probability $p^{(n)}_{1|0}=\mathrm{Tr}[\rho_0^{\otimes n}M_n]$.
Conversely, a type-II error or ``false negative'' occurs when Alice concludes that the state is $\rho_0^{\otimes n}$ when it actually is $\rho_1^{\otimes n}$, and has probability $p^{(n)}_{0|1}=\mathrm{Tr}[\rho_1^{\otimes n}(\mathbb{I}_n-M_n)]$.
From the quantum Stein's lemma \cite{hiai1991proper,ogawa2005strong}, the quantum relative entropy
\begin{equation}
S(\rho_1\|\rho_0) = \mathrm{Tr}\left[\rho_1\left(\ln\rho_1-\ln\rho_0\right)\right]
\end{equation}
governs the exponent of the asymptotic optimal decay of the probability of a false positive with a given probability of a false negative: for any $\epsilon>0$,
\begin{equation}\label{eq:limit}
\lim_{n\to\infty}\frac{1}{n}\ln\inf_{0\le M_n\le\mathbb{I}_n}\left\{p^{(n)}_{1|0}:p^{(n)}_{0|1}<\epsilon\right\} = -S(\rho_1\|\rho_0),
\end{equation}
such that for $n\to\infty$
\begin{equation}\label{eq:approx}
\inf_{0\le M_n\le\mathbb{I}_n}\left\{p^{(n)}_{1|0}:p^{(n)}_{0|1}<\epsilon\right\} \simeq \exp\left(-n\,S(\rho_1\|\rho_0)\right)\;.
\end{equation}

In the quantum illumination scenario, the natural choice of null hypothesis is ``target not present'', and the consequent choice of alternative hypothesis is ``target present''.
This setup has been analyzed in \cite{spedalieri2014asymmetric,zhuang2017entanglement,wilde2017gaussian} in the particular case of
Gaussian probes.
An $n$-mode Gaussian quantum system \cite{braunstein2005quantum,weedbrook2012gaussian,holevo2013quantum,serafini2017quantum} is the quantum system of $n$ harmonic oscillators, or $n$ modes of the electromagnetic radiation.
Let Alice hold the quantum state $\rho_{AM}$, where $A$ is the $n$-mode Gaussian quantum system of the signal sent towards the target, and $M$ is the memory quantum system that Alice keeps.
We stress that our analysis is not restricted to Gaussian quantum states, and $\rho_{AM}$ can be any state of the joint quantum system $AM$.
If the target is present, Alice gets back an attenuated signal with some thermal noise, and her final state is $\rho_{BM}=(\Phi\otimes\mathbb{I}_M)(\rho_{AM})$, where $\Phi$ is the quantum channel that models the noise.
Thermal Gaussian noise is the most common assumption and provides a faithful model in most practical scenarios \cite{tan2008quantum,wilde2017gaussian}.
Therefore, we assume $\Phi$ to be an $n$-mode noisy quantum Gaussian attenuator with thermal noise \cite{weedbrook2012gaussian,holevo2013quantum,serafini2017quantum}.
Let $a_i$ and $e_i$, $i=1,\ldots,n$, be the ladder operators of the system $A$ and the noise, respectively.
Alice then receives the system $B$ with ladder operators
\begin{equation}
b_i = \sqrt{\eta}\,a_i + \sqrt{1-\eta}\,e_i,\qquad i=1,\,\ldots,\,n,
\end{equation}
where $0<\eta\ll1$ is the fraction of the signal that is reflected (see \autoref{fig:figure1}).
Alice's goal is performing a measurement on her final state $\rho_{BM}$ to detect whether the target is there.
If the target is not present, Alice gets back only the noise, and her final state is $\omega_B\otimes \rho_M$, where $\omega_B=\Phi(|0\rangle\langle0|)$ is an $n$-mode thermal quantum Gaussian state, $|0\rangle$ is the $n$-mode vacuum and $\rho_M$ is the marginal on $M$ of $\rho_{AM}$.
This enforces that Alice cannot get any information on the presence of the target if she does not send any signal, i.e., if $\rho_{AM} = |0\rangle_A\langle0|\otimes\rho_M$.

If Alice could send an arbitrarily bright signal, she could always detect whether the target is present.
Therefore, we impose that the marginal state of the signal $\rho_A$ must have average number of photons at most $E>0$, i.e.
\begin{equation}\label{eq:constr}
\mathrm{Tr}_A\left[H_A\,\rho_A\right]\le n\,E,
\end{equation}
where $H_A$ is the standard Hamiltonian of $A$ that counts the total number of photons.
We do not consider the energy of the memory system.

\section{The optimal probes}
We consider the asymmetric discrimination in quantum illumination with $\rho_1 = \rho_{BM}$ and $\rho_0 = \omega_B\otimes \rho_M$.
The optimal asymptotic error exponent is governed by the following quantum relative entropy:
\begin{equation}\label{eq:S}
S(\rho_{BM}\|\omega_B\otimes \rho_M) = -S(B|M)_{\rho_{BM}} - \mathrm{Tr}_B\left[\rho_B\ln\omega_B\right],
\end{equation}
where
\begin{equation}
S(X|Y) = S(XY) - S(Y)
\end{equation}
is the quantum conditional entropy \cite{nielsen2010quantum,holevo2013quantum,hayashi2016quantum,wilde2017quantum} (see \cite{kuznetsova2010quantum,wilde2016energy} for the definition in the case where $S(XY) = S(Y) = \infty$ and \cite{pirandola2017fundamental} for a simple formula for Gaussian states).
If $M$ is not there, the relevant quantum relative entropy is
\begin{equation}\label{eq:S2}
S(\rho_{B}\|\omega_B) = -S(\Phi(\rho_A)) - \mathrm{Tr}_B\left[\rho_B\ln\omega_B\right].
\end{equation}

Our main result states that Alice's best probes for the asymmetric discrimination problem are the tensor powers of the two-mode squeezed vacuum states if she has a quantum memory, and the coherent states otherwise.
The coherent states of a one-mode Gaussian quantum system are the states of the form
\begin{equation}
|\alpha\rangle = \mathrm{e}^{-\frac{|\alpha|^2}{2}}\sum_{k=0}^\infty\frac{\alpha^k}{\sqrt{k!}}|k\rangle,\qquad \alpha\in\mathbb{C},
\end{equation}
where $\{|k\rangle\}_{k\in\mathbb{N}}$ is the Fock basis.
The coherent states of an $n$-mode Gaussian quantum system are the tensor products of the coherent states of each mode.
The two-mode squeezed vacuum states of a two-mode Gaussian quantum system are the states of the form
\begin{equation}\label{eq:tms}
|\psi(z)\rangle = \sqrt{1-z^2}\sum_{k=0}^\infty z^k\,|k\rangle\otimes|k\rangle,\qquad 0\le z<1.
\end{equation}
They can be obtained applying a two-mode squeezing to the two-mode vacuum state.

In most experimental setups we have $N_B\gg1$ and $E\ll1$, where $N_B$ and $E$ are the average number of photons per mode of the noise and of the probe, respectively \cite{genovese2016real,tan2008quantum,zhang2015entanglement}.
In this regime, the approximation \eqref{eq:approx} for the error probability is valid when the number of modes of the probe $n$ satisfies \cite{wilde2017gaussian}
\begin{equation}\label{eq:condn}
n \gg \frac{N_B}{\eta\,E}\;.
\end{equation}
Condition \eqref{eq:condn} is fulfilled by most experimental setups, e.g., in \cite{zhang2015entanglement} we have $n\approx 10^{11}$ with $\frac{N_B}{\eta\,E}\approx 10^7$.
Therefore, the limit of infinite modes in \eqref{eq:limit} does not restrict the applicability of our results to actual experiments.

\begin{thm}\label{thm:main}
Let $A$ be an $n$-mode Gaussian quantum system and $M$ a generic quantum system.
Then, for any joint quantum state $\rho_{AM}$ such that $\mathrm{Tr}_A[\rho_AH_A]\le n\,E$,
\begin{equation}\label{eq:S'}
S(\rho_{BM}\|\omega_B\otimes \rho_M) \le S(\sigma_{BA'}\|\omega_B\otimes \sigma_{A'}),
\end{equation}
where $A'$ is an $n$-mode Gaussian quantum system, $\sigma_{AA'}$ is the $n$-th tensor power of the two-mode squeezed vacuum state \eqref{eq:tms} such that its marginal $\sigma_A$ on $A$ has average number of photons per mode $E$, $\rho_{BM} = (\Phi\otimes\mathbb{I}_M)(\rho_{AM})$ and $\sigma_{BA'} = (\Phi\otimes\mathbb{I}_M)(\sigma_{AA'})$.
Therefore, the squeezed vacuum state $\sigma_{AA'}$ constitutes Alice's optimal strategy in the asymmetric discrimination scenario.
\end{thm}
\begin{thm}
Let $A$ be an $n$-mode Gaussian quantum system.
Then, for any quantum state $\rho_{A}$ of $A$ such that $\mathrm{Tr}_A[\rho_AH_A]\le n\,E$,
\begin{equation}\label{eq:S2'}
S(\Phi(\rho_A)\|\omega_B) \le S(\Phi(|\alpha\rangle\langle\alpha|)\|\omega_B),
\end{equation}
where $|\alpha\rangle$, $\alpha\in\mathbb{C}^n$ is an $n$-mode coherent state of $A$ with average number of photons per mode $E$ (i.e., with $|\alpha|^2=n\,E$).
Therefore, the coherent state $|\alpha\rangle$ constitutes Alice's optimal strategy in the asymmetric discrimination scenario when a quantum memory is not available.
\end{thm}
\begin{proof}
The term $- \mathrm{Tr}_B[\rho_B\ln\omega_B]$ in \eqref{eq:S} and \eqref{eq:S2} is a function of $\mathrm{Tr}_A[\rho_AH_A]$ alone.
Indeed, since $\omega_B$ is a thermal Gaussian state, $-\ln\omega_B = a\,H_B + b\,\mathbb{I}_B$ with $a,b>0$.
Hence, $- \mathrm{Tr}_B[\rho_B\ln\omega_B] = a\,\mathrm{Tr}_B[\rho_BH_B] + b$, and the claim follows since $\mathrm{Tr}_B[\rho_BH_B]$ is a linear increasing function of $\mathrm{Tr}_A[\rho_AH_A]$ \cite{holevo2013quantum}.

Since coherent states minimize the output entropy of the noisy quantum attenuators \cite{giovannetti2015solution,mari2014quantum}, Alice's best choice in \eqref{eq:S2} is choosing $\rho_A$ to be a $n$-mode coherent state with the maximum allowed average number of photons.
The best choice is not unique, since the average number of photons can be distributed in an arbitrary way among the $n$ modes.

The optimality of the two-mode squeezed vacuum states in \eqref{eq:S} follows from Theorem \ref{thm:S} below, a new fundamental entropic inequality for the noisy quantum Gaussian attenuators.
\end{proof}
\begin{thm}\label{thm:S}
Let $A$ and $B$ be $n$-mode Gaussian quantum systems and $M$ a generic quantum system, and let $\Phi:A\to B$ be a noisy quantum Gaussian attenuator with thermal noise.
Let $E\ge0$, and let $\rho_{AM}$ be a joint quantum state of $AM$ such that $\mathrm{Tr}_A[\rho_AH_A]\le n\,E$ and $S(\rho_M)<\infty$, where $H_A$ is the Hamiltonian of $A$, and $\rho_A$ and $\rho_M$ are the marginals of $\rho_{AM}$ on $A$ and $M$, respectively.
Then,
\begin{equation}
S(B|M)_{\rho_{BM}} \ge S(B|A')_{\sigma_{BA'}},
\end{equation}
where $A'$ is an $n$-mode Gaussian quantum system, $\sigma_{AA'}$ is the $n$-th tensor power of the two-mode squeezed vacuum state such that its marginal $\sigma_A$ on $A$ has average number of photons per mode $E$, $\rho_{BM} = (\Phi\otimes\mathbb{I}_M)(\rho_{AM})$ and $\sigma_{BA'} = (\Phi\otimes\mathbb{I}_{A'})(\sigma_{AA'})$.
\end{thm}
\begin{proof}
We can assume $\rho_{AM}$ pure.
Indeed, let $\rho_{AMR}$ be a purification of $\rho_{AM}$.
Since the marginals on $A$ of $\rho_{AM}$ and $\rho_{AMR}$ coincide, they have the same average number of photons.
The claim then follows from the strong subadditivity \cite{nielsen2010quantum,holevo2013quantum,hayashi2016quantum,wilde2017quantum}:
\begin{equation}
S(B|M)_{\rho_{BM}} \ge S(B|MR)_{{\rho}_{BMR}},
\end{equation}
where $\rho_{BMR}=(\Phi\otimes\mathbb{I}_{MR})(\rho_{AMR})$.

Let then $\rho_{AM}$ be a pure state.
We have
\begin{align}\label{eq:Stilde}
S(B|M)_{\rho_{BM}} &= S(\rho_{BM}) - S(\rho_M) = S(\tilde{\Phi}(\rho_A)) - S(\rho_A),\nonumber\\
S(B|A')_{\sigma_{BA'}} &= S(\tilde{\Phi}(\sigma_A)) - S(\sigma_A),
\end{align}
where $\tilde{\Phi}$ is the complementary channel of $\Phi$ \cite{nielsen2010quantum,holevo2013quantum,hayashi2016quantum,wilde2017quantum}, and $S(\rho_{BM})<\infty$ since $S(\rho_M)<\infty$ and $\rho_B$ has finite average number of photons \cite{holevo2013quantum}.
Let $\mathcal{N}:A\to A$ be the quantum additive noise channel such that $\mathrm{Tr}_A[H_A\mathcal{N}(\rho_A)] = n\,E$, and let $\rho'_{AM'}$ be a purification of $\rho'_A = \mathcal{N}(\rho_A)$.
Since $\tilde{\Phi}$ is a quantum Gaussian channel, we have from \cite[Lemma 10]{de2018conditional}
\begin{equation}
S(\tilde{\Phi}(\rho_A)) - S(\rho_A) \ge S(\tilde{\Phi}(\rho'_A)) - S(\rho'_A),
\end{equation}
hence
\begin{equation}
S(B|M)_{\rho_{BM}} \ge S(B|M')_{\rho'_{BM'}},
\end{equation}
where $\rho'_{BM'} = (\Phi\otimes\mathbb{I}_{M'})(\rho'_{AM'})$.
We can then assume $\mathrm{Tr}_A[H_A\rho_A] = n\,E$.

Now $\sigma_A$ is the thermal Gaussian state with the same average number of photons as $\rho_A$, and \cite{holevo2013quantum}
\begin{align}
S(\rho_A\|\sigma_A) &= S(\sigma_A) - S(\rho_A),\nonumber\\
S(\tilde{\Phi}(\rho_A)\|\tilde{\Phi}(\sigma_A)) &= S(\tilde{\Phi}(\sigma_A)) - S(\tilde{\Phi}(\rho_A)),
\end{align}
where the second line follows from Lemma \ref{lem:energy}.
The claim then follows from the data processing inequality for the relative entropy \cite{nielsen2010quantum,holevo2013quantum,hayashi2016quantum,wilde2017quantum}:
\begin{align}
&S(B|M)_{\rho_{BM}} - S(B|A')_{\sigma_{BA'}}\nonumber\\
&= S(\rho_A\|\sigma_A) - S(\tilde{\Phi}(\rho_A)\|\tilde{\Phi}(\sigma_A)) \ge0.
\end{align}
\end{proof}

\begin{lem}\label{lem:energy}
Let $\tilde{\Phi}$ be the complementary channel of the $n$-mode noisy quantum Gaussian attenuator with thermal noise.
Then, for any $n$-mode quantum state $\rho_A$
\begin{equation}\label{eq:claimlem}
S(\tilde{\Phi}(\rho_A)\|\tilde{\Phi}(\sigma_A)) = S(\tilde{\Phi}(\sigma_A)) - S(\tilde{\Phi}(\rho_A)),
\end{equation}
where $\sigma_A$ is the $n$-mode thermal Gaussian state with the same average number of photons as $\rho_A$.
\end{lem}
\begin{proof}
We will prove that
\begin{equation}\label{eq:PhiH}
\tilde{\Phi}^\dag(\ln\tilde{\Phi}(\sigma_A)) = x\,H_A + y\,\mathbb{I}_A,
\end{equation}
where $x,\,y\in\mathbb{R}$, $\tilde{\Phi}^\dag$ is the dual channel of $\tilde{\Phi}$ \cite{holevo2013quantum} defined by
\begin{equation}
\mathrm{Tr}\left[X\,\tilde{\Phi}(Y)\right] = \mathrm{Tr}\left[\tilde{\Phi}^\dag(X)\,Y\right]
\end{equation}
for any bounded operator $X$ and any trace class operator $Y$, and $H_A = \sum_{i=1}^n a_i^\dag a_i$ is the Hamiltonian on $A$, where the $a_i$ are the ladder operators associated to the $n$ modes of $A$.
The claim \eqref{eq:claimlem} then follows since
\begin{align}
S(\tilde{\Phi}(\rho_A)\|\tilde{\Phi}(\sigma_A)) &= S(\tilde{\Phi}(\sigma_A)) - S(\tilde{\Phi}(\rho_A))\nonumber\\
&\phantom{=} + \mathrm{Tr}_A\left[\left(\sigma_A-\rho_A\right)\tilde{\Phi}^\dag(\ln\tilde{\Phi}(\sigma_A))\right],
\end{align}
and $\rho_A$ and $\sigma_A$ have the same average number of photons.

Since both $\sigma_A$ and $\tilde{\Phi}$ are the $n$-th tensor power of the corresponding one-mode versions, the left-hand side of \eqref{eq:PhiH} is the sum of one term for each of the $n$ modes.
Then, it is sufficient to prove the claim \eqref{eq:PhiH} for $n=1$, when $H_A=a^\dag a$.
Since $\tilde{\Phi}$ is a Gaussian channel, $\tilde{\Phi}(\sigma_A)$ is still a Gaussian state, $\ln\tilde{\Phi}(\sigma_A)$ is a quadratic polynomial in the quadratures, and $\tilde{\Phi}^\dag(\ln\tilde{\Phi}(\sigma_A))$ is a quadratic polynomial in $a$ and $a^\dag$ \cite{holevo2013quantum}.
The claim \eqref{eq:PhiH} follows if we prove that
\begin{equation}\label{eq:claim}
\mathrm{e}^{\mathrm{i}H_At}\,\tilde{\Phi}^\dag(\ln\tilde{\Phi}(\sigma_A))\,e^{-\mathrm{i}H_At} = \tilde{\Phi}^\dag(\ln\tilde{\Phi}(\sigma_A))\qquad\forall\;t\in\mathbb{R}.
\end{equation}
Indeed, $a^\dag a$ and $\mathbb{I}$ are the only invariant quadratic polynomials since $\mathrm{e}^{\mathrm{i}H_At}\,a\,\mathrm{e}^{-\mathrm{i}H_At} = \mathrm{e}^{-\mathrm{i}t}\,a$.

Let us now prove \eqref{eq:claim}.
The complementary channel of the one-mode noisy attenuator is \cite{holevo2013quantum}
\begin{equation}
\tilde{\Phi}(\rho_A) = \mathrm{Tr}_A\left[\left(U_{AE}\otimes\mathbb{I}_{E'}\right)\left(\rho_A\otimes\gamma_{EE'}\right)\left(U_{AE}^\dag\otimes\mathbb{I}_{E'}\right)\right],
\end{equation}
where $\gamma_{EE'}$ is a two-mode squeezed vacuum state of the one-mode Gaussian quantum systems $E$ and $E'$, and $U_{AE}$ implements a beam-splitter on $AE$.
We have for any $t\in\mathbb{R}$
\begin{align}
\mathrm{e}^{-\mathrm{i}\left(H_E-H_{E'}\right)t}\,\gamma_{EE'}\,\mathrm{e}^{\mathrm{i}\left(H_E-H_{E'}\right)t} &= \gamma_{EE'},\nonumber\\
\mathrm{e}^{-\mathrm{i}\left(H_A+H_E\right)t}\,U_{AE}\,\mathrm{e}^{\mathrm{i}\left(H_A+H_E\right)t} &= U_{AE},
\end{align}
where $H_E$ and $H_{E'}$ are the Hamiltonians of $E$ and $E'$, respectively.
Then, for any quantum state $\rho_A$, any bounded operator $X$ and any $t\in\mathbb{R}$
\begin{align}
&\tilde{\Phi}\left(\mathrm{e}^{-\mathrm{i}H_At}\,\rho_A\,\mathrm{e}^{\mathrm{i}H_At}\right) = \mathrm{e}^{-\mathrm{i}\left(H_E-H_{E'}\right)t}\,\tilde{\Phi}(\rho_A)\,\mathrm{e}^{\mathrm{i}\left(H_E-H_{E'}\right)t},\nonumber\\
&\tilde{\Phi}^\dag\left(\mathrm{e}^{\mathrm{i}\left(H_E-H_{E'}\right)t}\,X\,\mathrm{e}^{-\mathrm{i}\left(H_E-H_{E'}\right)t}\right) = \mathrm{e}^{\mathrm{i}H_At}\,\tilde{\Phi}^\dag(X)\,\mathrm{e}^{-\mathrm{i}H_At},
\end{align}
and the claim \eqref{eq:PhiH} follows.
\end{proof}

\section{Discussion}
We have proven that the tensor powers of the two-mode squeezed vacuum states are Alice's best probes in the asymmetric discrimination problem of quantum illumination, in the sense that they minimize the decay rate of the probability of a false positive with a given probability of a false negative.
This fundamental striking result implies that nor correlations nor entanglement among the modes of the signal that Alice sends can decrease the error probability.
Alice's best strategy is then the simplest one, and nothing can be gained by using any more exotic probe, such as states with multi-mode entanglement.

Putting together our results with \cite{cooney2016strong} (valid for finite-dimensional systems), the optimality of the two-mode squeezed vacuum states might be extended to the adaptive scenario with feedback \cite{hentschel2010machine}.
One potential approach is combining our results with the simulation and reduction techniques of \cite{pirandola2017ultimate}.
Another potential approach might be extending \cite{cooney2016strong} to infinite dimension exploiting the sandwiched R\'enyi divergences \cite{berta2016r}.

The main open question left is whether the optimality of the two-mode squeezed vacuum states is limited to the asymmetric discrimination.
We conjecture that this is not the case, and the two-mode squeezed vacuum states are optimal also in the symmetric discrimination problem.
Proving this conjecture will be the subject of future work.

\section*{Acknowledgements}
We thank Matthias Christandl for very helpful discussions and Mark Wilde and Stefano Pirandola for a careful reading of the paper and useful comments.

GdP acknowledges financial support from the European Research Council (ERC Grant Agreements no 337603 and 321029), the Danish Council for Independent Research (Sapere Aude),VILLUM FONDEN via the QMATH Centre of Excellence (Grant No. 10059), and the Marie Sk\l odowska-Curie Action GENIUS (grant no. 792557).
JB acknowledges financial support from the European Research Council (ERC Grant Agreements no 337603), the Danish Council for Independent Research (Sapere Aude), VILLUM FONDEN via the QMATH Centre of Excellence (Grant No. 10059), and Qubiz - Quantum Innovation Center.

\includegraphics[width=0.05\textwidth]{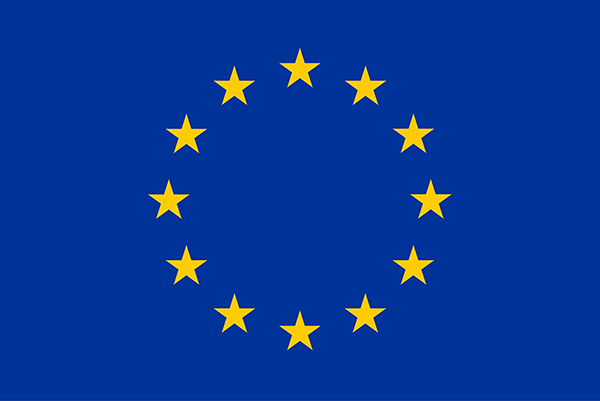}
This project has received funding from the European Union's Horizon 2020 research and innovation programme under the Marie Sk\l odowska-Curie grant agreement No. 792557.

\bibliographystyle{apsrev4-1}
\bibliography{ill}
\end{document}